\documentclass[11pt,letterpaper]{article}
\usepackage[margin=1in]{geometry}
\usepackage{amsthm,amsmath,amssymb}
\usepackage[dvipsnames]{xcolor}
\usepackage{tikz}
\usepackage{graphicx} % Required for inserting images
\usepackage[font=small,labelfont=bf]{caption} % For captions
\usepackage[normalem]{ulem}
\usepackage{enumerate}
\usepackage[T1]{fontenc}
\usepackage{soul} % For underlines that wrap 
\usepackage{algorithm}
\usepackage{algpseudocode}
\usepackage{hhline}
\usepackage{subcaption}
\usepackage{multirow}
\usepackage{xspace}
\usepackage{hyperref}
\usepackage{thm-restate}

\newtheorem{theorem}{Theorem}[section]
\newtheorem{corollary}[theorem]{Corollary}
\newtheorem{lemma}[theorem]{Lemma}
\newtheorem{claim}[theorem]{Claim}
\theoremstyle{definition}

\newcommand{\ouralg}{\textsc{TripodTracker}\xspace}
\newcommand{\bridge}{B}

\newcommand{\new}{\mathrm{new}}
\newcommand{\OPT}{\mathrm{OPT}}
\newcommand{\cost}{\mathrm{COST}}
\newcommand{\off}{\mathrm{off}}
\newcommand{\on}{\mathrm{on}}
\newcommand{\org}{\mathrm{org}}
\newcommand{\rel}{\mathrm{rel}}
\newcommand{\dis}[1]{#1}

\algnewcommand\algorithmicforeach{\textbf{for each}}
\algdef{S}[FOR]{ForEach}[1]{\algorithmicforeach\ #1\ \algorithmicdo}
\algnewcommand{\LineComment}[1]{\State \(\triangleright\) #1}

\title{Online 3-Taxi on General Metrics}
 \author{Christian Coester\thanks{Funded by the European Union (ERC, CCOO, 101165139). Views and opinions expressed are however those of the author(s) only and do not necessarily reflect those of the European Union or the European Research Council. Neither the European Union nor the granting authority can be held responsible for them.}\\ University of Oxford \and Tze-Yang Poon\\ University of Oxford}
\date{}

\begin{document}

\maketitle

\begin{abstract}
    The online $k$-taxi problem, introduced in 1990 by Fiat, Rabani and Ravid, is a generalization of the $k$-server problem where $k$ taxis must serve a sequence of requests in a metric space. Each request is a pair of two points, representing the pick-up and drop-off location of a passenger. In the interesting ``hard'' version of the problem, the cost is the total distance that the taxis travel without a passenger.
    
    The problem is known to be substantially harder than the $k$-server problem, and prior to this work even for $k=3$ taxis it has been unknown whether a finite competitive ratio is achievable on general metric spaces. We present an $O(1)$-competitive algorithm for the $3$-taxi problem. 
\end{abstract}

\pagenumbering{gobble} %DO NOT REMOVE. The title page needs to be a page on its own and doesn't count towards the page count
\newpage
\pagenumbering{arabic}

\section{Introduction}

The $k$-taxi problem, originally proposed by Karloff and formalized by Fiat, Rabani and Ravid in 1990~\cite{FiatRR90}, is a fundamental online problem that generalizes the $k$-server problem by associating each request with a source and destination. In this problem, $k$ taxis move in a metric space and must serve a sequence of requests arriving online. Each request specifies a pick-up location and a drop-off location. To serve it, a taxi must move first to the pick-up point and then to the drop-off point. The taxi serving each request must be chosen by an online algorithm without knowledge of future requests.

There are two versions of the problem, known as the ``easy'' and ``hard'' $k$-taxi problem, which differ in how the cost is defined: In the easy $k$-taxi problem, the cost is the total distance traveled by all taxis. In the hard $k$-taxi problem, the cost is defined as only the total \emph{overhead} distance of empty runs; that is, distances traveled to get to pick-up locations count towards the cost, but distances traveled from pick-up to drop-off (which are the same regardless of algorithm) are excluded from the cost. Although the optimal (offline) solutions are the same for both models, the smaller costs in the hard version mean that the competitive ratio between the cost of an online algorithm and an optimal (offline) solution is higher. Indeed, Coester and Koutsoupias \cite{CoesterK19} showed that the easy $k$-taxi problem is exactly equivalent to the $k$-server problem, with a deterministic competitive ratio between $k$ and $2k-1$, whereas the hard version is at least exponentially harder, with a lower bound of $\Omega(2^k)$ for deterministic algorithms. Therefore, research has focused on the hard version, and all mentions of the $k$-taxi problem hereafter refer to the hard version unless stated otherwise.

In terms of algorithms, prior to this work, a finite competitive ratio has been known to be achievable for general metric spaces only for the case of $k=2$ taxis, where the deterministic competitive ratio is exactly $9$~\cite{CoesterK19}. Additional results exist for special metric spaces: an $O(1)$-competitive algorithm for three taxis on a line metric \cite{CoesterK19}, $O(2^k)$-competitive algorithms for ultrametrics \cite{CoesterK19,BuchbinderCN23}, and an $O(k^D)$-competitive algorithm for weighted trees of combinatorial depth $D$~\cite{BuchbinderCN23}. Using randomization, for $n$-point metric spaces with aspect ratio $\Delta$ there exist multiple different algorithms with the following competitive ratios: $O(2^k\log n)$ based on the aforementioned result for ultrametrics~\cite{CoesterK19}, $O((n\log k)^2\log n)$ based on an algorithm for a more general problem of ``metrical service systems with transformations'' \cite{BubeckBCS21}, $2^{O(\sqrt{\log k \log\Delta})}\log n$ based on a reverse-time primal-dual analysis of the Double Coverage algorithm on ultrametrics \cite{BuchbinderCN23}, and most recently $O\left(\log^3 \Delta\cdot \log^2(nk\Delta)\right)$ based on a new linear programming relaxation for the problem \cite{GuptaKP24}. Note, however, that the latter collection of bounds is vacuous on general metric spaces where the number of points $n$ and $\Delta$ could be infinite.

Despite the interest that the problem has generated, for general metric spaces (and even seemingly simple cases such as three taxis on the two-dimensional Euclidean plane) it has remained unknown since the problem's introduction 35 years ago whether any finite competitive ratio is achievable when $k>2$. In this paper, we provide a positive answer for $k=3$.

\begin{theorem}
    There exists an $O(1)$-competitive deterministic online algorithm for the $3$-taxi problem on general metric spaces.
\end{theorem}

\subsection*{Additional Related Work} A competitive algorithm for the hard $2$-taxi problem due to Karloff has been known since its introduction (see \cite{FiatRR90}), and \cite{FiatRR90} also gave a first algorithm for the \emph{easy} $k$-taxi problem by adapting their algorithm for the $k$-server problem. A stochastic version of the (easy) $k$-taxi problem was studied in~\cite{DehghaniEHLS17}.

The $k$-server problem corresponds to the special case of the $k$-taxi problem where for each request, the pick-up point is equal to the drop-off point. Its competitive ratio is $\Theta(k)$ deterministically~\cite{ManasseMS88,KoutsoupiasP95}. For randomized algorithms, there are polylog$(k,n)$- and polylog$(k,\Delta)$-competitive algorithms~\cite{BansalBMN15,BubeckCLLM18} and a lower bound of $\Omega(\log^2 k)$~\cite{BubeckCR23}, which is also the best lower bound for randomized algorithms for the $k$-taxi problem.

Besides generalizing the $k$-server problem, \cite{CoesterK19} showed that the (deterministic) $k$-taxi problem is also a generalization of the width-$k$ layered graph traversal problem, which is also equivalent to chasing sets of cardinality $k$ in a metric space~\cite{FiatFKRRV98,BubeckCR22}. The aforementioned $\Omega(2^k)$ lower bound on the $k$-taxi problem is inherited from the same lower bound on these problems.

\section{Preliminaries}

Let $(S,d)$ be a metric space. To simplify notation, for two points $x,y\in S$, we will often write $\dis{xy}:=d(x,y)$ for their distance. A \emph{configuration} is a multiset of $k$ points in $S$, representing the locations of $k$ taxis. For two configurations $C$ and $C'$, we denote by $d(C,C')$ the cost of a minimum weight perfect matching between them. This captures the total distance traveled to move taxis from $C$ to $C'$.

An instance of the $k$-taxi problem on a metric space $(S,d)$ consists of an initial configuration $C_0$ and a sequence $(r_1,s_1),(r_2,s_2),\dots,(r_T,s_T)$ of requests, each of which is a pair of two points in $S$. An algorithm is said to serve the request sequence if it outputs a sequence of configurations $\hat{C}_{1},\hat{C}_{2},\dots,\hat{C}_{T}$ such that for all $t\in\{1,\dots, T\}$, $r_t\in\hat{C}_{t}$. After the algorithm reaches configuration $\hat C_t$, the taxi at $r_t$ serves the request $(r_t,s_t)$ by relocating to $s_t$. This changes the configuration to $C_{t}=\hat{C}_{t}-r_t+s_t$, where the $+$ and $-$ operators add/remove one copy of a point from a configuration. The cost of the algorithm is defined as $\sum_{t=1}^{T}d(C_{t-1},\hat{C}_t)$.

An online algorithm must choose each configuration $\hat C_t$ after the request $(r_t,s_t)$ is revealed and without knowledge of future requests. In contrast, an offline algorithm knows the request sequence in advance and can therefore serve it optimally. We denote by $\cost$ and $\OPT$ the costs of an online algorithm and the optimal offline algorithm, respectively. The online algorithm is $\rho$-competitive if $\cost\leq \rho \cdot\OPT + c$ for all request sequences, where $c$ is a constant that may depend only on the metric space and the initial configuration, but not the request sequence. 

\paragraph{Bridges and Tripods.} We may assume without loss of generality that for any two points $x,y\in S$, the metric space contains a continuous \emph{bridge} $B(x,y)$ between $x$ and $y$. That is, $B(x,y)\subseteq S$ is isometric to a closed interval of length $\dis{xy}$ whose endpoints are $x$ and $y$. This can be achieved by adding virtual points to the metric space. We describe our algorithm in a way that it may move taxis to these virtual points. To turn this into an algorithm on the original metric space (without virtual points), we can defer the movement of any taxi until it actually serves a request, which always happens at non-virtual points. By the triangle inequality, deferring movement cannot increase the cost of the algorithm.

Similarly, we may assume without loss of generality that for any three points $x,y,z\in S$, the metric space contains a continuous \emph{tripod} $B(x,y,z)$. That is, there is a point $e\in S$ (depending on $x,y,z$) such that $\dis{xy}=\dis{xe}+\dis{ye}$, $\dis{xz}=\dis{xe}+\dis{ez}$ and $\dis{yz}=\dis{ye}+\dis{ez}$, and the tripod $B(x,y,z)$ is the union of the three bridges $B(x,e)$, $B(y,e)$ and $B(z,e)$. See Figure~\ref{fig:tripod}. The point $e$ is called the \emph{center} or \emph{branching point} of the tripod $B(x,y,z)$ and can be added to the metric space by connecting it to $x$, $y$ and $z$ by edges of lengths
\begin{align}
    &\dis{xe}=\frac{\dis{xy}+\dis{xz}-\dis{yz}}{2}, \quad \dis{ye}=\frac{\dis{xy}-\dis{xz}+\dis{yz}}{2}, \quad ze=\frac{-\dis{xy}+\dis{xz}+\dis{yz}}{2}.\label{eq:tripod}
\end{align}
Note that the union of any two edges of $B(x,y,z)$ yields a bridge between two of its endpoints. In general, there may be multiple bridges/tripods for a given pair $(x,y)$ or triple $(x,y,z)$, and we use $B(x,y)$ and $B(x,y,z)$ to refer to any one of them chosen arbitrarily (unless further specified).

\begin{lemma} \label{lemma:deformBridges}
    Given points $u,v,w_1,w_2$, let $h_1$ and $h_2$ be the branching points of $\bridge{}(u,v,w_1)$ and $\bridge{}(u,v,w_2)$, respectively. Then any two corresponding edges (sharing the same endpoint $u$ or $v$, and likewise the edges $w_1h_1$ and $w_2h_2$) differ in length by at most $\dis{w_1w_2}$.
\end{lemma}
\begin{proof}
    This follows from equations~\eqref{eq:tripod} by the triangle inequality.
\end{proof}

\begin{figure}
    \centering
    \includegraphics{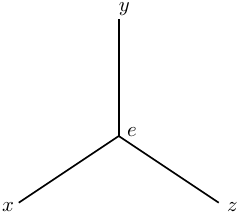}
    \caption{The tripod $B(x,y,z)$ with center $e$.}\label{fig:tripod}
\end{figure}

\paragraph{Active, passive and unobstructed taxis.} At any time, we call the taxi that served the previous request the \textit{active} taxi and the other taxis \textit{passive}. If there was no previous request, an arbitrary taxi is considered active. In our $3$-taxi algorithm, we denote the active online taxi by $x_1$ and the two passive taxis $x_2$ and $x_3$, reindexing the taxis between requests appropriately. We use $x_i$ to refer to both the taxi as well as its location.

When the current request is $(r,s)$, we call a passive taxi \textit{unobstructed} if the shortest path between it and $r$ on the tripod $\bridge(x_2,x_3,r)$ does not contain the location of the other passive taxi. Equivalently, a passive taxi is unobstructed unless the other passive taxi is at the branching point of $B(x_2,x_3,r)$. In the degenerate case where both passive taxis are at the same location, we regard one of them as unobstructed, chosen arbitrarily.

% \begin{observation} \label{obs:finiteMetrics}
% \begin{enumerate} [(i)]
%     \item Every 2-point metric is isometric to a single edge.
%     \item Every 3-point metric comprising the points $x,y,z$ is isometric to a 3-pointed star with branching point $e$ satisfying the following equations:
%     \begin{align*}
%         &\dis{xe}=\frac{\dis{xy}+\dis{xz}-\dis{yz}}{2}, \quad \dis{ye}=\frac{\dis{xy}-\dis{xz}+\dis{yz}}{2}, \quad ze=\frac{-\dis{xy}+\dis{xz}+\dis{yz}}{2}.
%     \end{align*}
%     \item Every 2-point metric can be extended to a 3-point metric by adding the appropriate branching point along the edge $\dis{xy}$ and an edge $ze$ of the appropriate length.
% \end{enumerate}
% \end{observation}

% \begin{definition} [2,3-point bridges] \label{defn:234pointMetrics}
%     \begin{enumerate} [(i)]
%         \item For any two points $x,y\in X$, the 2-point bridge $\bridge{}(x,y)$ comprises a path of length $\dis{xy}$.
%         \item For any three points $x,y,z\in X$, the 3-point bridge $\bridge{}(x,y,z)$ comprises the three-pointed star with branching point $e$ satisfying the equations in Observation \ref{obs:finiteMetrics}.
%     \end{enumerate}
% \end{definition}

% \subsection{Overview of Techniques.}

% \begin{itemize}
%     \item forgiveness.
% \end{itemize}

\section{An Algorithm for 3-Taxi on General Metrics}

We will define a deterministic online algorithm \ouralg{} for the 3-taxi problem on general metric spaces. Although it suffices to move a single taxi in response to a request, it is more convenient for the description and analysis of \ouralg to simultaneously move several taxis continuously when a new request arrives, similarly to the \textsc{DoubleCoverage} algorithm for $k$-server~\cite{ChrobakKPV91,ChrobakL91}. Our algorithm builds on ideas from the algorithm of~\cite{CoesterK19} for three taxis on the line metric, but generalizing it to arbitrary metrics requires overcoming new structural challenges. We adapt both the algorithm and the potential function used for its analysis in ways that are crucial to handle general metrics. As a result, even when specialized to the line metric, our algorithm and potential function differ from~\cite{CoesterK19}. Notably, our approach also streamlines certain aspects, avoiding the need for a separate treatment of seven different movement cases as in~\cite{CoesterK19}.

\subsection{Overview}
When a request $(r,s)$ arrives, \ouralg{} proceeds by moving the taxis simultaneously at different speeds towards $r$, until a taxi reaches $r$. The active taxi moves along the bridge $\bridge(x_1,r)$, while the passive taxis move along the tripod $\bridge(x_2,x_3,r)$. The active taxi plays a special role because at the start of the request we can always guarantee that there is an offline taxi at the same location. Thus, it is preferable to keep a taxi nearby, and accordingly \ouralg{} moves the active taxi at a much slower speed than the passive taxis. 

The movement speeds of the two passive taxis depend on the structure of $\bridge(x_2,x_3,r)$. \ouralg{} uses $\bridge(x_2,x_3,r)$ to decide which of the two passive taxis is better suited to serve the request. Consider a scenario where $x_2$ is located at the center of  the tripod $\bridge(x_2,x_3,r)$ and thus on a shortest path from $x_3$ to $r$. Then there would be no reason to serve the request using $x_3$, since the online algorithm could instead serve the current request with $x_2$ and defer the movement of $x_3$ to the original location of $x_2$ to a later request. As such, if one of the passive taxis is at the branching point of $\bridge(x_2,x_3,r)$, we move only that passive taxi and the active taxi towards the request point $r$. This is similar to the algorithm for 2-taxi on general metrics. Otherwise both passive taxis are unobstructed. In this case, it is less clear which passive taxi is better suited to serve the request, so we simultaneously move both passive taxis towards $r$ and towards each other by moving them towards the branching point. %This is similar to the  %Our potential functions and passive taxi movement speeds are constructed to enable this movement, but the potential functions differ from those of \textsc{DoubleCoverage} because they have to be invariant under the relocation action.

Inspired by the algorithm for the 3-taxi problem on the line~\cite{CoesterK19}, we keep track of two disjoint ``intervals", each starting at a passive taxi and ending at some point along the bridge $\bridge{}(x_2,x_3)$ that is part of the tripod $B(x_2,x_3,r)$. (Recall that there may be several bridges $B(x_2,x_3)$, so when a new request appears, we map these intervals to a possibly different bridge $B(x_2,x_3)$ that is formed by two edges of the tripod with the new request.) Intuitively, each interval mark a region where the passive taxi located at its endpoint holds ``more responsibility'' than the other passive taxi. Accordingly, the intervals are defined only along bridges between the passive taxis: outside these bridges, the passive taxis never move simultaneously and therefore do not need to be distinguished by responsibility. The active taxi does not have an interval and is treated specially because, unlike in $k$-server, the active taxi can be relocated to a new location at no cost, making it volatile and ``unfit for holding responsibility'' beyond its current location.

These intervals -- together with the role of the active taxi -- encode the algorithm's entire memory about the past. This information is used by the algorithm to determine the passive taxi movement speeds, as well as by the potential functions to relate online movement costs to offline movement costs. How these intervals are updated, how they influence movement speeds, and how they affect the potential function all differ from the previous algorithm for the line metric~\cite{CoesterK19}.

\subsection{Intuition}

In order to make sense of the details of \ouralg{}'s behavior, it will be helpful to briefly describe part of the potential function we will use in the analysis. This potential will be the minimum weight of a certain \emph{distorted} matching between the online and offline configurations. The active online taxi $x_1$ will always be matched to the active offline taxi, and this pair will contribute its undistorted distance to the potential. The pairs involving the passive online taxis $x_2$ and $x_3$ may contribute less than their actual distance and their matched partners will be selected to minimize the total weight of the matching. 

In the scenario where the active offline taxi moves to serve the current request, we will charge our movement costs directly to the offline movement cost. This is possible because the total distance moved by the active online taxi is no greater than that of the offline active taxi, and all online movement is at most a constant multiple of the active online taxi's moved distance. 

Otherwise, we gain the additional assumption that the active offline taxi did not move. In this case, we instead charge our movement costs to a net decrease in the aforementioned potential. To achieve this, we will first assume that the offline algorithm moves a taxi to $r$ before the online algorithm moves, such that during \ouralg{}'s movement at least one of the passive online taxis is getting closer to its matched partner at $r$ and thus decreasing its matching potential contribution. By moving the active taxi at a sufficiently slow speed, we ensure that its increase in matching contribution is smaller than the above decrease.
%This results in a decrease in the matching contribution of that taxi's pair in the matching, which will be designed to be larger than any increases in the contributions of the other two taxis' pairs.
%The contribution of each of the three pairs of taxis to the matching potential will be determined by carefully selected \emph{discount factors}. As the pair involving the active online taxi is undistorted, its discount factor is $1$, corresponding to no discount, whereas the two pairs involving the passive online taxis are discounted to factors less than or equal to $1$.
%The difference between the discount factors of the active and passive taxis is linked to the slower movement speed of the active taxi: any movement of the active taxi towards its matched partner results in a large decrease in the potential, so moving at a slow speed is sufficient to compensate for the other two taxis. 
However, in the movement cases where there are two passive taxis moving, we will need to ensure that the passive taxi moving towards its matched partner at $r$ can decrease the potential more than the other passive taxi increases it by possibly moving away from its matched partner. We achieve this by designing the potential to assign a discount factor to specific portions of the path from each passive taxi to its matched partner, and accordingly these portions will contribute less to the potential. The assignments of this discount is implicitly tracked by the aforementioned intervals and the structure of several tripods, and is described in Section \ref{section:analysis}. See Figure~\ref{fig:discountAssignments} for an example, where $1-\psi<1$ is the discount factor.

\ouralg{} makes use of the differing discounts in two possible ways. If it is able to decrease the amount of undiscounted distance between a passive taxi and $r$, this means the distorted distance of that pair decreases relatively quickly, which we leverage in the analysis. Alternatively, if the entire distance between a passive taxi and $r$ is already discounted, then moving this taxi towards $r$ decreases the distorted distance more slowly. In this case, \ouralg{} moves that passive taxi towards $r$ at a slightly faster speed (i.e., $1+b$ instead of $1$), so that the distorted distance to $r$ decreases at a comparable rate. %This is similar in principle to the slow movement speed of the active taxi: the stronger the discount on the distance that we are removing, the faster we have to move that taxi. 

A major challenge is that \ouralg{} does not actually know which of the three online taxis is matched to the offline taxi at $r$, as this depends on the unknown locations of the other offline taxis. This is the main reason for the significant care required in the choice of movement speeds and discount factors. As we will see later, there are also cases where the movement cost cannot be charged to the matching potential, and we will also employ an additional potential function for this purpose.

The first online taxi which reaches the request location $r$ will be used to serve the request. If one of the passive taxis is the first to reach $r$, it will additionally become the new active taxi. However, this results in a reassignment of matching partners, which could cause an increase in the overall matching potential due to an increase in the amount of undiscounted distance included in the matching. \ouralg{} handles this by reorganizing the intervals, having the newly passive taxi \emph{inherit} some of the interval that previously belonged to the newly active taxi that just served the current request and giving the passive taxi that remains passive an additonal amount of interval to reflect its ``additional responsibility" relative to the newly passive taxi. Through these reorganizations, \ouralg{} ensures that any overall change in potential can be charged to the offline movement in the current request.

\subsection{Algorithm Description}

A pseudocode of our algorithm~\ouralg is provided in Algorithm~\ref{alg:main}, and an example of it serving a request is depicted in Appendix \ref{sec:example}. We keep track of an interval for each passive taxi $x_i$, whose first endpoint is the passive taxi location $x_i$ and the other endpoint is denoted by $q_i$. We use $\ell_i:=\dis{x_iq_i}$ for the length of this interval.\footnote{In our analysis, each passive taxi will enjoy a discount in any portion of a path that belongs to its own interval.} Before the first request arrives, we initialize both interval lengths $\ell_i$ to $0$ in line~\ref{line:initialize} of the algorithm.

When a request $(r,s)$ arrives, the main task is to move a taxi to the pick-up location $r$.
The interval endpoint $q_i$ is chosen as the point at distance $\ell_i$ from $x_i$ on the path between the two passive taxis that is part of the tripod $B(x_2,x_3,r)$. We will maintain throughout the run of the algorithm that the two intervals of the passive taxis are interior-disjoint, i.e., $\ell_2+\ell_3\le \dis{x_2x_3}$.

Next (lines~\ref{line:movementStart}-\ref{line:movementEnd}), we continuously and simultaneously move the active taxi and each unobstructed passive taxi towards $r$, until one of them reaches $r$. The active taxi moves at some small speed $a\in(0,1)$, to be determined later, along a bridge from its old location to $r$. Each unobstructed passive taxi $x_i$ moves towards $r$ along the tripod $B(x_2,x_3,r)$. We denote by $e$ the branching point of this tripod. Note that $e$ remains unchanged while both passive taxis are unobstructed, and if there is a single unobstructed passive taxi, then $e$ is at the same location as this taxi and they move together towards $r$. The movement speed of each unobstructed passive taxi $x_i$ depends on whether the branching point $e$ belongs to this taxi's interval or not: if it is inside the interval, but not at $x_i$, then the movement is at a fast speed of $1+b$, for some constant $b>0$ to be determined later. Otherwise the movement is at speed $1$. Note that the case where $\dis{x_ie}=0$ is precisely the case where $x_i$ is the only unobstructed taxi. At the same time, we also move the associated interval endpoint $q_i$ towards the branching point $e$ at speed 1 if it is not at $e$ already, and update $\ell_i$ to maintain that it is the length of the interval between $x_i$ and $q_i$. In the case where $e$ itself is moving (at speed 1 towards $r$, because it is at the location of the single unobstructed passive taxi), this means that the distance between $q_i$ and $e$ remains unchanged, since $q_i$ and $e$ then both move towards $r$ on the path from $q_i$ through $x_i=e$ to $r$. If a passive taxi reaches point $e$ during the movement, the other taxi is no longer unobstructed and stops moving.

\begin{algorithm}
\caption{\ouralg}\label{alg:main}
\begin{algorithmic}[1]
    \Require Initial taxi location $x_1,x_2,x_3$
    \State $(\ell_2,\ell_3)\gets (0,0)$ \label{line:initialize}
    \ForEach{request $(r,s)$}
        \For{$i\in\{2,3\}$}
            \State $q_i\gets$ point on $B(x_2,x_3)\subseteq B(x_2,x_3,r)$ with $\ell_i=\dis{x_iq_i}$.
        \EndFor
        \While{$r\not\in \{x_1,x_2,x_3\}$}\label{line:movementStart}
            \State move $x_1$ along $\bridge{}(x_1,r)$ towards $r$ at speed $a$
            \State $e\gets$ branching point of $\bridge{}(x_2,x_3,r)$
            \ForEach{unobstructed passive taxi $x_i$}
                \State move $x_i$ along $B(x_2,x_3,r)$ towards $r$ at speed 
                    $\begin{cases}
                        1+b\quad &\text{if $\ell_i\ge \dis{x_ie}>0$,}\\
                        1\quad&\text{otherwise}
                    \end{cases}$
                %\If{$q_i\ne e$}
                    \State move $q_i$ towards $e$ at speed $1$ (if it is not already at $e$)
                %\EndIf
                \State $\ell_i\gets \dis{x_iq_i}$
            \EndFor
        \EndWhile \label{line:movementEnd}
        \If{passive taxi at $r$} \label{line:reorganizationStart}
            \Comment{WLOG $x_3=r$}
            \State $f\gets$ the branching point of $\bridge(x_1,x_2,x_3)$

            \State $\ell_2\gets \min(\ell_2+ \dis{x_1f},\dis{x_1x_2})$ \label{line:l2}
            \State $\ell_3\gets \max(0,\ell_3-\dis{x_1x_3})$ \label{line:l3}

            \State reindex $(x_1,x_3)\gets (x_3,x_1)$

        \EndIf \label{line:reorganizationEnd}
        \State serve the request with $x_1$\label{line:serveRequest}
    \EndFor
\end{algorithmic}
\end{algorithm}

Once a taxi reaches the pick-up point $r$, all taxis stop moving. If the taxi reaching $r$ is passive, some reorganization is necessary as the roles of active and passive taxis change, and we need to update the interval lengths to ensure that the intervals around the new passive taxis can again be placed in an interior-disjoint way on a bridge between the new passive taxis, and that the potential used in the analysis does not increase. The details of this reorganization are described in lines~\ref{line:reorganizationStart}-\ref{line:reorganizationEnd} for the case that $x_3$ is the passive taxi reaching the request. The other case is symmetric. Intuitively, the interval length $\ell_2$ of the passive taxi that remains passive increases, as it is in some sense ``more passive'' than the newly passive taxi, and is then truncated to at most the new distance $x_1x_2$ between the two passive taxis. The active taxi that becomes passive inherits the interval length $\ell_3$ of the passive taxi that becomes active, but reduced by the distance between these two taxis, and truncated at $0$.

Finally, the new active taxi serves the request by moving from $r$ to $s$.

The following claim establishes the aforementioned invariant that the two intervals around the passive taxis remain interior-disjoint.

\begin{restatable}{claim}{intervalInvariant}\label{cl:intervalInvariant}
    At all times during the continuous movement, the interval endpoints appear on the path from $x_2$ to $x_3$ in the order $x_2\to q_2\to q_3\to x_3$, possibly with equality between consecutive points in the sequence.
\end{restatable}

\begin{proof}
    By definition, we maintain $\ell_i=\dis{x_iq_i}$ at all times, so the invariant is true following the initialization in line \ref{line:initialize}. During the continuous movement, if a single passive taxi $x_i$ is unobstructed, then $\dis{x_ie}=0$ and $x_i$ and $q_i$ travel away from the endpoints of the other interval at the same speed, maintaining the invariant. If both passive taxis are unobstructed, then all four interval endpoints move towards $e$ (except some $q_i$ which may be at $e$ already). If all movements are at speed $1$, then this clearly maintains the invariant. The case where a passive taxi moves at speed $1+b$ also doesn't violate the invariant, as this only occurs if there is a positive gap $\ell_i>0$ between $x_i$ and $q_i$.

    It remains to show that after line \ref{line:l3} we have $\ell_2^\new+\ell_3^\new\le \dis{x_1x_2}$, where $\ell_i^\new$ denotes the value after the update. Since $\dis{x_1x_2}$ will be the new distance between the passive taxis, this will ensure the invariant holds for the next request. Clearly the statement is true if $\ell_3^\new=0$. Otherwise, we get
    \begin{align*}
        \ell_2^\new+\ell_3^\new\le \ell_2+\dis{x_1f}+\ell_3-\dis{x_1x_3} = \ell_2+\ell_3-\dis{x_3f} \le \dis{x_2x_3}-\dis{x_3f} = \dis{x_2f}\le \dis{x_1x_2},
    \end{align*}
    where we used that $\ell_2+\ell_3\le \dis{x_2x_3}$ was true since the invariant held previously.
\end{proof}

\subsection{Analysis} \label{section:analysis}

We refer to the three offline taxis as $y_1,y_2,y_3$. Let $\cost$ and $\OPT$ be the costs incurred by \ouralg{} and the optimal offline algorithm, respectively, in serving a given request sequence. Let $\cost_t$ and $\OPT_t$ be the costs incurred in serving the $t^{\text{th}}$ request by the online and offline algorithms, respectively. Let $\Phi_t$ be the value of a non-negative potential function after both algorithms have served the $t^{\text{th}}$ request and let $\Delta\Phi_t = \Phi_t-\Phi_{t-1}$. To show that \ouralg{} is $\kappa$-competitive for some constant $\kappa$, it suffices to show that for all requests 
\begin{equation} \label{eq:condition}
    \cost_t + \Delta\Phi_t \leq \kappa\OPT_t,
\end{equation}
if we additionally require that $\Phi_0=0$. 

For each request, we assume that the offline algorithm moves a taxi to $r$ first, followed by the online algorithm. This allows us to guarantee that during the online algorithm's movement, there is an offline taxi located at $r$. We split each request into four phases: offline taxi movement, online taxi movement (lines \ref{line:movementStart}-\ref{line:movementEnd}), reorganization (lines \ref{line:reorganizationStart}-\ref{line:reorganizationEnd}) and relocation (line \ref{line:serveRequest}). Let the change in potential during these phases for the $t^{\text{th}}$ request be $\Delta\Phi_t^\off{}, \Delta\Phi_t^\on{}$, $\Delta\Phi_t^\org{}$ and $\Delta\Phi_t^\rel{}$, respectively. Then online cost $\cost_t$ is only incurred during the online movement phase and offline cost $\OPT_t$ is only incurred during the offline movement phase. We consider the former two phases in Section \ref{section:movements} and the latter two in Section \ref{section:reorgAndReloc}.

We use a potential of the form $\Phi = \alpha M + \beta\Sigma$, where $\alpha>\beta>1$ are constants and $M$ and $\Sigma$ are two components of the potential.

To define these components, we first generalize the definition of intervals around the passive taxis: Recall that in the algorithm, we used $q_i$ as a point on the tripod $B(x_2,x_3,r)$, and by Claim~\ref{cl:intervalInvariant} they always reside on the path between $x_2$ and $x_3$. Since this path embeds isometrically into any other tripod $B(x_2,x_3,z)$, we can define points $q_i$ on the $(x_2,x_3)$-path in any such tripod as well, at the same relative distance from $x_2$ and $x_3$ (i.e., such that $\ell_i=\dis{x_iq_i}$). For convenience, we will denote these points by $q_i$ again.

The component $M$ is the minimum weight of the aforementioned distorted matching between the online and offline taxi locations. We restrict the active online and offline taxis (denoted $x_1$ and $y_1$) to always be matched to each other, and require that this pair contributes its actual (undistorted) distance to $M$. The pairs involving the passive online taxis $x_2$ and $x_3$ may contribute less than their actual distance. The actual contribution of a pair $(x_i,y_i)$, where $x_i$ is a passive online taxi and $y_i$ the offline taxi it is matched to, is determined by the structure of the tripod $B(x_2,x_3,y_i)$, including the location of $q_2,q_3$ on this tripod. Denote by $g_i$ the branching point of $B(x_2,x_3,y_i)$. Let the unique $(x_i,y_i)$-path on $\bridge(x_2,x_3,y_i)$ be $P_i$. The two passive offline taxis $y_2,y_3$ are indexed to minimize the expression
\begin{equation*}
     M= \dis{x_1y_1} + \int_{P_2}f_2(z)\,dz + \int_{P_3}f_{3}(z)\, dz,
\end{equation*}
where
\begin{equation*}
    f_i(z) = \begin{cases}
        1-\psi &\text{if }x_iz\leq \ell_i\text{ or }x_iz\geq \dis{x_ig_i}, \\
        1 & \text{otherwise,}
    \end{cases}
\end{equation*}
for some constant $0<\psi<1$. In other words, portions of $P_i$ in $x_i$'s own interval and on the $y_i$ edge of $\bridge(x_2,x_3,y_i)$ are discounted to the smaller factor of $1-\psi$ and the remaining portion is not discounted and contributes its entire distance to the matching. See Figure \ref{fig:discountAssignments}.

\begin{figure}
    \centering
    \includegraphics{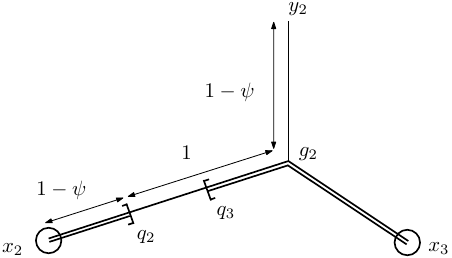}
    \caption{The matching contribution of the $(x_2,y_2)$-pair is determined by the undiscounted factor of $1$ and the discounted factor of $1-\psi$ over the different portions of the path $P_2$ from $x_2$ to $y_2$.} 
    \label{fig:discountAssignments}
\end{figure}

The component $\Sigma$ of the potential is defined as
\begin{align*}
    \Sigma = \max(0,\min(\ell_2,\ell_3)-\dis{q_2q_3}).
\end{align*}
The $\Sigma$ component plays a less prominent role than $M$ in the proof, and we need it to pay for the algorithm's movement only in the case where both interval endpoints $q_2$ and $q_3$ are located at the branching point $e$ of $B(x_2,x_3,r)$, as the $M$ potential might not decrease in this case. The coefficient $\beta$ of $\Sigma$ in the overall potential will be much smaller than the coefficient $\alpha$ of $M$, so that any adverse change to $\Sigma$ in other cases will be negligible compared to the change in $M$. The $\Sigma$ potential plays a similar role to the sum of pairwise server distances potential used in the analysis of the \textsc{DoubleCoverage} algorithm for $k$ servers. However, it has been altered so that it does not increase under relocation requests that can bring taxis arbitrarily far apart.

The constants $a$ and $b$ determining the movement speeds of the algorithm and the constants $0<\psi<1<\beta<\alpha$ used in this analysis are chosen\footnote{We can express all constants in terms of a sufficiently small positive constant $0<\epsilon\ll 1$. In ascending order, the constants are $a=\epsilon^4, b=\epsilon^2,  \psi=\epsilon,\beta=\epsilon^{-2},\alpha=\epsilon^{-5} $.} to obey the following hierarchy:
\begin{align}
\label{eq:constants}
     a \ll b \ll \psi \ll 1 \ll \alpha a \ll \beta\ll \alpha b.
\end{align}
Terms higher in the hierarchy are assumed to be much larger than those lower in the hierarchy. Hence, to demonstrate that $A\leq 0$ for some expression $A$, as long as the coefficients of all terms are bounded, it suffices to show that some term in $A$ is $<0$ and all higher order terms in $A$ are $\leq 0$.

Our analysis will repeatedly make use of the following lemma to restrict how the passive taxis can be matched to each other in the minimum matching. Concretely, if we consider the two tripods each with endpoints comprising the two passive online taxis and one of the offline taxis, Lemma \ref{lemma:branchingPointOrder} states that in the minimum matching each offline taxi is matched to the online taxi which is relatively closer to the branching point of the corresponding tripod. See Figure \ref{fig:branchingPointOrder}. This generalizes the intuitive property from on a line metric that a minimum matching would match online to offline taxis in left-to-right order (e.g., the leftmost online taxi to the leftmost offline taxi etc.), and it holds even in the presence of discount intervals.

\begin{figure}
    \centering
    \includegraphics{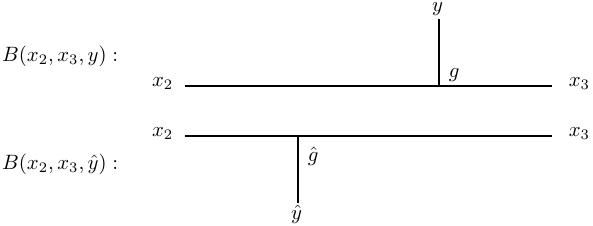}
    \caption{When the two passive online taxis $x_2$ and $x_3$ are matched to two offline taxis $y$ and $\hat{y}$, the relative distance of the branching points $g$ and $\hat{g}$ to each passive taxi determines the minimum matching. In the minimum matching of the above scenario, $x_2$ is matched to $\hat{y}$ and $x_3$ is matched to $y$.} 
    \label{fig:branchingPointOrder}
\end{figure}

\begin{lemma} \label{lemma:branchingPointOrder}
    Let $y,\hat{y}$ be the two passive offline taxis. Let $g$ and $\hat{g}$ be the branching points of $\bridge{}(x_2,x_3,y)$ and $\bridge{}(x_2,x_3,\hat{y})$, respectively. Then $x_2$ is matched to $\hat{y}$ and $x_3$ is matched to $y$ in a minimum matching iff $\dis{x_2\hat{g}}\leq \dis{x_2g}$.
\end{lemma} 

\begin{proof}
    The only other matching to consider is when $x_3$ is matched to $\hat{y}$ and $x_2$ to $y$. The discount factors over the distance from each online/offline taxi to the nearer branching point (i.e. $\min(\dis{x_3g}, \dis{x_3\hat{g}})$, $\min(\dis{x_2g}, \dis{x_2\hat{g}})$,  $\dis{yg}$ and $\dis{\hat{y}\hat{g}}$) are unchanged, so these distances have the same contribution to the cost of either matching. Hence, the only change in matching contribution from swapping from the matching in the lemma statement to this other matching is due to the change in the distance from each taxi to its matched partner by  $\dis{x_2g}-\dis{x_2\hat{g}}$. Therefore, this distance is non-negative iff the matching in the lemma statement is a minimum matching.
\end{proof}

\subsubsection{Online and Offline Movement}
\label{section:movements}

In this section we will show that 
\begin{equation}\label{eq:movementCondition}
    \cost_t+\Delta\Phi_t^{\off}+\Delta\Phi_t^{\on{}}\leq \kappa\OPT_t.
\end{equation}

Without loss of generality, we can assume that the offline algorithm only moves one taxi when serving each request, since delaying any other movements to a later request does not increase its overall cost.

\begin{lemma} \label{lemma:offlineMovement}
    During the offline movement phase it holds that  $\Delta\Phi_t^\off{}\leq \alpha \OPT_t$.
\end{lemma}
\begin{proof}
    Since there are no changes to the online taxi locations or $q_2,q_3$, we have that $\Delta\Sigma_t^{\off{}}=0$. Considering the offline taxi moved by the offline algorithm, the change in the distance to its matched online taxi is no greater than its moved distance, due to the triangle inequality. Furthermore, since $M$ is a weighted sum of distances of the minimal matching where the weights are all no greater than $1$, we have that $\Delta M_t^\off{} \leq \OPT_t$, which  suffices to give the claim.
\end{proof}

The case where the offline algorithm serves the current request with the active offline taxi $y_1$ is simple. Then $\OPT_t=\dis{y_1r}=\dis{x_1r}$ is an upper bound on the distance moved by the active online taxi. The distances moved by the passive online taxis are at most a constant factor larger. Furthermore, $\Delta\Phi^\on{}$ is no more than $O(\alpha+\beta)$ times the total distance moved by all online taxis. Hence, using Lemma~\ref{lemma:offlineMovement} we get that $\cost_t+\Delta\Phi_t^{\off}+\Delta\Phi_t^{\on{}}$ is at most a constant factor larger than $\OPT_t$, so for $\kappa$ sufficiently large this suffices to show \eqref{eq:movementCondition}.

This leaves us in the case where the offline algorithm moves a passive taxi. Therefore, $y_1$ does not move and there is a passive offline taxi at $r$. Our goal is to show the following claim.

\begin{claim} \label{claim:onlineMovement}
    If the offline algorithm moves a passive taxi, then during the online movement phase,
    \begin{equation}\label{eq:onlineCondition}
        \cost_t + \Delta\Phi_t^\on{} \leq 0.
    \end{equation}
\end{claim}

Together, Lemma \ref{lemma:offlineMovement} and Claim \ref{claim:onlineMovement} suffice to prove \eqref{eq:movementCondition} in the case where the offline algorithm moves a passive taxi. Using the following corollary on the matched partners of the passive taxis, we will show that the passive online taxi that is not matched to $r$ (and therefore could be moving away from its matched partner) only increases its matching contribution at a rate of $(1-\psi)$.

\begin{lemma} \label{lemma:bridgeChanges}
     Let $y$ be any point and let the rates of change of $\dis{x_2y},\dis{x_3y},\dis{x_2x_3}$ be $\dis{x_2y}',\dis{x_3y}',\dis{x_2x_3}'$ respectively. Let $g$ be the branching point of $\bridge{}(x_2,x_3,y)$. Then the rates of change of $\dis{x_2g},\dis{x_3g},\dis{yg}$ are respectively
    \begin{align*}
        \dis{x_2g}' = \frac{\dis{x_2y}'-\dis{x_3y}'+\dis{x_2x_3}'}{2}, \quad  \dis{x_3g}' = \frac{- \dis{x_2y}'+\dis{x_3y}'+\dis{x_2x_3}'}{2}, \quad   \dis{yg}' = \frac{\dis{x_2y}'+\dis{x_3y}'-\dis{x_2x_3}'}{2}.
    \end{align*}
\end{lemma}
\begin{proof}
    This follows from \eqref{eq:tripod}.
\end{proof}

\begin{corollary} \label{cor:passiveMovingAwayMatching}
    In the movement cases with two unobstructed passive taxis, if one of the passive taxis $x_i$ is matched to $r$ and the other passive taxi $x_j$ moves at speed $c_j$, then $x_j$'s matching contribution increases at a rate of at most $c_j(1-\psi)$.  
\end{corollary}
\begin{proof}
    Let $c_i$ be the movement speed of $x_i$, let $x_j$ be matched to $y$ and let $e,g$ be the branching points of $\bridge{}(x_i,x_j,r),\bridge{}(x_i,x_j,y)$ respectively. 

    Due to Lemma \ref{lemma:bridgeChanges}, since $x_i$ and $x_j$ move towards each other, the total length of the segment $\dis{x_jg}$ in $\bridge{}(x_i,x_j,y)$ does not increase. Then it suffices to show that the length of the undiscounted region for $x_j$ in the segment $\dis{x_jg}$ does not increase. This length is only positive when $\dis{x_jq_j}<\dis{x_jg}$. Then by Lemma \ref{lemma:branchingPointOrder}, we also have that $\dis{x_jq_j}<\dis{x_je}$. In this case, both $x_j$ and $q_j$ move towards $e$ at the same speed, so the distance $\dis{x_jq_j}$ does not change. Thus, the length of the undiscounted region on $\dis{x_jg}$ cannot increase.
\end{proof}

We can now prove Claim \ref{claim:onlineMovement} by showing that it holds over all movement cases.

% \begin{lemma} \label{lemma:onlineMovement}
%     (Proof of Claim \ref{claim:onlineMovement}.) If the offline algorithm moves a passive taxi, then over lines \ref{line:movementStart}-\ref{line:movementEnd} it holds that $\cost_t + \Delta\Phi_t^\on{} \leq 0$.
% \end{lemma}
\begin{proof}[Proof of Claim \ref{claim:onlineMovement}]
    From ordering \eqref{eq:constants}, we have the following hierarchy on terms which we will use in this proof:
    \begin{align} \label{eq:movementConstants}
        \alpha a \ll \beta \ll \alpha b \ll \alpha \psi \ll \alpha .
    \end{align}
    Let $\cost'$ and $\Phi'$ be the sum of the movement speeds of all online taxis and the rate of change of $\Phi$, respectively. We show that in all cases, $\cost'+\Phi'\leq 0 $. Technically $\cost'$ and $\Phi'$ are not defined when \ouralg{} switches from one movement case to another, but these changes happen finitely often and do not affect the proof. In all cases, the movement of $x_1$ away from its matched partner contributes $a$ to $ M'$, so it contributes $\alpha a$ to $\Phi'$. In all movement cases, we will show that some term strictly higher than $\alpha a$ in the ordering \eqref{eq:movementConstants} is $<0$ and all higher terms are $\leq 0$. Hence, we can ignore $\cost'$ and $x_1$'s contribution to the matching for the rest of the analysis.

    In all movement cases with two unobstructed passive taxis, we have that $\Sigma'\leq2$. We now analyze $\Phi'$, using Corollary \ref{cor:passiveMovingAwayMatching} to identify the matching contribution of the passive taxi not matched to $r$.

    \begin{itemize}
        \item If the passive taxi $x_i$ matched to $r$ satisfies $\dis{x_iq_i}<\dis{x_ie}$, it contributes $\leq-1$ to $ M'$. Then the other passive taxi contributes $\leq (1+b)(1-\psi)$ to $ M'$. Hence, the $\alpha\psi$ term in $\Phi'$ is $<0$ and all higher terms in ordering \eqref{eq:movementConstants} are $\leq0$.
        \item Suppose the passive taxi $x_i$ matched to $r$ satisfies $\dis{x_iq_i}\geq \dis{x_ie}$ and the other passive taxi $x_j$ satisfies $\dis{x_jq_j}< \dis{x_je}$. Then $x_i$ contributes $-(1+b)(1-\psi)$ to $ M'$ and $x_j$ contributes $\leq (1-\psi)$. Hence, the $\alpha b$ term in $\Phi'$ is $<0$ and all higher terms in ordering \eqref{eq:movementConstants} are $\leq0$.
        \item Otherwise, the passive taxi $x_i$ matched to $r$ satisfies $\dis{x_iq_i}= \dis{x_ie}$ and the other passive taxi $x_j$ satisfies $\dis{x_jq_j}= \dis{x_je}$. Then $x_i$ contributes $-(1+b)(1-\psi)$ to $ M'$ and $x_j$ contributes $\leq(1+b)(1-\psi)$. Furthermore, $\Sigma'=-(1+b)<0$ and all terms higher than $\beta$ in ordering \eqref{eq:movementConstants} are $\leq 0$.
    \end{itemize}

When there is only one unobstructed passive taxi, by Lemma \ref{lemma:branchingPointOrder}, the unobstructed taxi is moving towards its matched partner. Hence, the $\alpha$ term of $\Phi'$ is $<0$, so by ordering \eqref{eq:movementConstants}, the claim $\cost'+\Phi'\leq 0 $ holds.
\end{proof}

\subsubsection{Reorganization and Relocation}
\label{section:reorgAndReloc}

For the remaining phases, the online and offline taxis no longer move. Hence, we will show that $\Delta\Phi^{\org}_t \leq \alpha\cdot\psi\cdot\OPT_t$ and $\Delta\Phi^{\rel}_t \leq 0$, which, together with the previous section, suffices to show \eqref{eq:condition} for sufficiently large $\kappa$. To avoid ambiguity, we consider the reindexing of the offline active taxi to happen in the reorganization phase.

We begin with the reorganization phase.

\begin{lemma}
    During reorganization it holds that $\Delta\Sigma^{\org{}}_t\leq0$.
\end{lemma}
\begin{proof}
    If $x_1=r$, no changes take place during reorganization, so $\Delta\Sigma^{\org}_t=0$. Otherwise, without loss of generality $x_3=r$. If $\ell_3\leq \dis{x_1x_3}$, then $\ell_3^\new=0$ after the reorganization and $\Sigma$ goes to 0. Otherwise, both $q_2$ and $q_3$ are on the $x_2$ edge of $B(x_1,x_2,x_3)$ and $\dis{q_2f} \ge \dis{q_3f}> \dis{x_1f}$ (as in Figure~\ref{fig:example}). Prior to reorganization, $\Sigma=\max(0,\min(\ell_2,\ell_3)-\dis{q_2q_3})$. %If $fq_3<\dis{x_1f}$ then $\ell_3^\new=0$, so $\Sigma$ goes to 0. Otherwise, $fq_3\geq \dis{x_1f}$.
    After the reorganization, the distance $\dis{q_2q_3}$ increases by $\dis{x_1f}$. If $\Sigma$ was determined by $\ell_3$ prior to the reorganization, $\ell_3^\new\leq \ell_3$ so $\Delta\Sigma^{\org}_t\leq0$. Otherwise $\Sigma$ was determined by $\ell_2$ prior to the reorganization and $\ell_2^\new{}\leq \ell_2+\dis{x_1f}$, so $\Delta\Sigma^{\org}_t=0$.
\end{proof}

\begin{lemma} \label{lemma:nonMovementCases}
    During reorganization, it holds that $\Delta M^{\org}_t\leq \psi\cdot\OPT_t$.
\end{lemma}

\begin{proof}
    We compare the matching potentials before and after reorganization. We consider four cases, defined by whether the online/offline algorithm serves the request with an active/passive taxi. In all cases except the one where an online passive taxi and the offline active taxi serve the request, we will prove the stronger claim that $\Delta M^{\org}_t\leq 0$.
    Note that reorganization involves reindexing the active online/offline taxis to the taxis at $r$, which may change the set of matchings we consider. 

    If both active taxis serve the request, no changes in the matching happen during reorganization so $\Delta M^{\org}_t=0$.

    Suppose that the active online taxi $x_1$ and a passive offline taxi $y_3$ serve the request. It suffices to show that reorganizing $y_3$ to be the new active offline taxi by matching $x_1$ to $y_3$ and $x_3$ to $y_1$ does not increase the matching potential. Then the matching contribution of the $(x_2,y_2)$-pair does not change. The matching contribution of the $(x_1,y_1)$-pair was $x_1y_1$ and this decreases to 0 in the $(x_1,y_3)$-pair. When changing $x_3$'s partner from $y_3$ to $y_1$, its matching contribution will be defined on the tripod $B(x_2,x_3,y_1)$ instead of $B(x_2,x_3,y_3)$. By the triangle inequality, the distance from $x_3$ to its matched partner increases by at most $\dis{y_3y_1}=\dis{x_1y_1}$. Furthermore, by Lemma \ref{lemma:deformBridges}, the corresponding $x_3$-edges of these tripods differ in length by at most $\dis{x_1y_1}$. Hence, $x_3$'s matching contribution increases by at most $\dis{x_1y_1}$, since in the worst case the entire increase in $x_3$'s distance to its matched partner is incurred on the $x_3$ edge of the corresponding tripods and is discounted at rate 1 because the amount of $x_3$'s interval used in the matching does not increase. Therefore, the decrease in $x_1$'s matching contribution is greater than the increase in $x_3$'s, so $\Delta M^{\org}_t\leq 0$.

    In the remaining cases, the passive online taxi $x_3$ serves the request. We begin by showing that $x_2$'s matched partner is never the offline taxi that will become active.

    \begin{claim} \label{claim:x2NotMatchedTor}
        Before reorganization, there exists a minimum matching where either $x_1$ or $x_3$ is matched to $r$.
    \end{claim} 

    \begin{proof}
        If the offline algorithm served the current request with the active taxi $y_1$, then $x_1$ is matched to $r$. Otherwise, one of the two passive taxis is matched to $r$. We can apply Lemma~\ref{lemma:branchingPointOrder} with $y=r$. Since $x_3=r=e$, the lemma shows that there is a minimum matching where $x_3$ is matched to $r$.
    \end{proof}

    Therefore, despite reindexing the active taxi, it is always valid for $x_2$'s partner $y_2$ to be unchanged after reorganization and for $y_2$ to remain passive. 

    \begin{claim}
        If $x_3$ serves the request and $x_2$'s matched partner does not change, then $x_2$'s matching contribution does not increase after reorganization.
    \end{claim}
    \begin{proof}
        Recall that $x_2$ is matched to $y_2$ before reorganization. We show that the cost of the $(x_2,y_2)$-pair in $\bridge{}(x_2,x_1,y_2)$ after reorganization is no larger than the cost of the pair in $\bridge{}(x_2,x_3,y_2)$ before reorganization. Let $f,g,h$ be the branching points of $\bridge{}(x_1,x_2,x_3)$, $\bridge{}(x_2,x_3,y_2)$ and $\bridge{}(x_2,x_1,y_2)$, respectively.
        It suffices to show that the amount of $(1-\psi)$-discounted region in the $(x_2,y_2)$ matching contribution does not decrease. This is true because either the entire path from $x_2$ to $y_2$ is discounted after reorganization or the size of the of the discounted region increases by at least $\dis{x_1f}$ on the $x_2$-edge and changes by $\dis{y_2h}-\dis{y_2g}$ on the $y_2$-edge. The total change in the amount of discounted region is thus at least $\dis{x_1f}+\dis{y_2h}-\dis{y_2g}=\dis{x_1f}+ \dis{x_2g}-\dis{x_2h}$. From \eqref{eq:tripod}, we have
        \begin{align*}
            \dis{x_1f} = \frac{1}{2}(\dis{x_1x_2}+\dis{x_1x_3}-\dis{x_2x_3}), \ \  \dis{x_2g} = \frac{1}{2}(\dis{x_2x_3} + \dis{x_2y_2}-\dis{x_3y_2}),\ \  
            \dis{x_2h} = \frac{1}{2}(\dis{x_1x_2}+\dis{x_2y_2}-\dis{x_1y_2}).
        \end{align*}
        Hence,
        \begin{align*}
            \dis{x_1f} + \dis{x_2g}- \dis{x_2h} &=\frac{1}{2}(\dis{x_1x_3}-\dis{x_3y_2}+\dis{x_1y_2})\geq 0,
        \end{align*}
        which suffices to show that $x_2$'s matching contribution does not increase during reorganization.
    \end{proof}

    It remains to show that $x_1$ and $x_3$'s matching contributions do not increase too much. Recall that the active and the remaining passive offline taxis are $y_1$ and $y_3$, respectively.

    If passive taxis $x_3$ and $y_3$ serve the request, then $x_3$'s matching contribution is $0$ before and after reorganization. Furthermore, $x_1$'s discount factor over its entire matching only changes from $1$ to at most $1$, which does not increase its matching contribution.

    Otherwise, $x_3$ and $y_1$ serve the request. We will show that matching the online taxi at $r$ ($x_3=x_1^\new{}$) to $r$ and the other online taxi $(x_1=x_3^\new{})$ to $y_3$ after reorganization increases the matching contribution by no more than $\psi\cdot \dis{x_1x_3}$. This suffices to prove the lemma since $\dis{x_1x_3}=x_1r\leq \OPT_t$, as the offline taxi $y_1$ moved from the \emph{old} location of $x_1$ to $r$, and the new location of $x_1$ can only be closer to $r$. The matching contribution of the $(x_1,r)$-pair was $\dis{x_1x_3}$ and this decreases to $0$ in the $(x_1^\new{},r)$-pair. Therefore, it suffices to show that the matching contribution of the new $(x_3^\new{},y_3)$-pair exceeds that of the old $(x_3,y_3)$-pair by at most $(1+\psi)\cdot\dis{x_1x_3}$. By the triangle inequality the total increase in the distance of this pair is at most $x_1x_3$ and by Lemma \ref{lemma:deformBridges} the length of the $x_3$-edge of $\bridge{}(x_2,x_3,y_3)$ increases by at most $\dis{x_1x_3}$ to give the $x_1$-edge in $\bridge{}(x_2,x_1,y_3)$. This contributes an increase of at most $x_1x_3$ to the matching potential, since in the worst case the entire increase in distance is incurred on what was originally the $x_3$-edge and is undiscounted. Due to line \ref{line:l3}, $\ell_3$ also decreases by at most $\dis{x_1x_3}$, contributing an additional $\psi\cdot x_1x_3$ to the matching potential. Hence, the total increase in matching potential is at most $(1+\psi)\cdot\dis{x_1x_3}$, giving the lemma.
\end{proof}

Finally, we consider the relocation phase.

\begin{lemma}
    During line \ref{line:serveRequest} it holds that $\Delta\Phi^{\rel}_t \leq 0$.
\end{lemma}

\begin{proof}
    There are no changes to the passive taxi locations or $q_2,q_3$, so $\Sigma$ is unchanged. The active taxis continue to share the same location, so
    $M$ is also unchanged.
\end{proof}

This completes the proof of our main theorem.

\begin{theorem}
    \ouralg{} is a $\kappa$-competitive algorithm for the hard 3-taxi problem.
\end{theorem}

\section{\texorpdfstring{Conclusion and $\boldsymbol{k>3}$}{Conclusion and k>3}} \label{section:Conclusion}

Our result shows that competitive algorithms for the $k$-taxi problem on general metrics exist beyond the previous barrier of $k=2$. The obvious open question is whether our result can be further extended to general $k$, with a competitive ratio depending only on $k$. We make the following two observations. 

First, our proof continues to make use of the idea of distinguishing the active taxi from the passive taxi(s) as per the previous proof for $k=2$, but extends on this by further distinguishing the two passive taxis with intervals. Understanding these intervals more deeply (and beyond our current level of understanding) seems crucial for extending the result to general $k$. The rest of this paragraph recapitulates the authors' current understanding.
%Note that the total amount of interval can only increase in the reorganization phase, when the interval of the passive taxi that remains passive increases. In all other steps, the interval lengths can only decrease. These intervals can therefore be interpreted as a method of distinguishing ``relative levels of passivity" between the two passive taxis, with the ``more passive" taxi denoted by having a larger interval. This taxi's matching contribution would generally be assigned a stronger discount factor, increasing its chances to move faster. 
As alluded to in our intuition section, we interpret them as marking regions where one passive taxi holds ``more responsibility'' than the other passive taxi. Note that the intervals are fully specified by their lengths $\ell_i$ and the locations of the passive taxis, so we can view $\ell_i$ as the ``responsibility score'' of taxi $x_i$. Recall that we interpret the active taxi as special and ``unfit for holding responsibility''. Accordingly, the only time when some $\ell_i$ can grow is in line 17 of the algorithm, when the previously active taxi $x_1$ becomes passive: Taxi $x_2$ (which was passive before and remains passive) used to be more responsible than $x_1$ (which was active and now becomes passive). To record this, we increase the responsibility score $\ell_2$ by the part of the distance from $x_2$ to $x_1$ that was outside any bridge $B(x_2,x_3)$ (i.e., $x_1f$ in line 17). Similarly, when a passive taxi becomes active, it loses its responsibility, and the responsibility is inherited by the newly passive taxi (except it is reduced by the distance between these two taxis, corresponding to the fact that the interval start point moves by this distance; line 18).
For $k=3$, we need intervals only to distinguish the relative responsibility between the two passive taxis. For $k>3$, a generalization of our approach might involve distinguishing responsibility levels between any pair of passive taxis. In fact, it might be more natural to encode the special role of the active taxi (and its lack of responsibility) also through such intervals.

Second, tripods currently exactly capture the distances between their three endpoints, and are used in this proof to dynamically embed the two passive taxis and any third point into a tree structure. Hierarchical Separated Trees (HSTs) are a natural candidate to replace tripods when $k>3$. We suspect that our algorithm may have an alternative interpretation in terms of dynamic embeddings into HSTs. Specifically, intervals with small discount factors, which stimulate faster movement, suggest that their endpoints are embedded to nearby points in the HST. Our distorted matching potential then corresponds to a minimum matching with respect to HST distances. Understanding our algorithm through such a lens may give insights into how such an HST embedding should evolve dynamically. We note that dynamic HST embeddings have been used successfully for the $k$-server problem \cite{BubeckCLLM18}.

A useful intermediate step would be to consider $k$ taxis on the line metric. On the line, all tripods have at least one edge of length 0, and any generalizations of tripods to higher $k$'s are similarly more restricted. 

Finally, we hope that our techniques can inspire progress on other variants of the $k$-server problem where existing results are limited to $k=2$ or restricted metrics, such as the weighted $k$-server problem and the generalized $k$-server problem~\cite{FiatR94,SittersS06,Sitters14,BansalEK17,BienkowskiJS19,ChiplunkarV20,AyyadevaraC21,BansalEKN23,BijoyMC26}.

\newpage
\appendix

\section{Example of \ouralg{} Serving a Request} \label{sec:example}

~

\begin{figure} [H]
    \centering
    \includegraphics{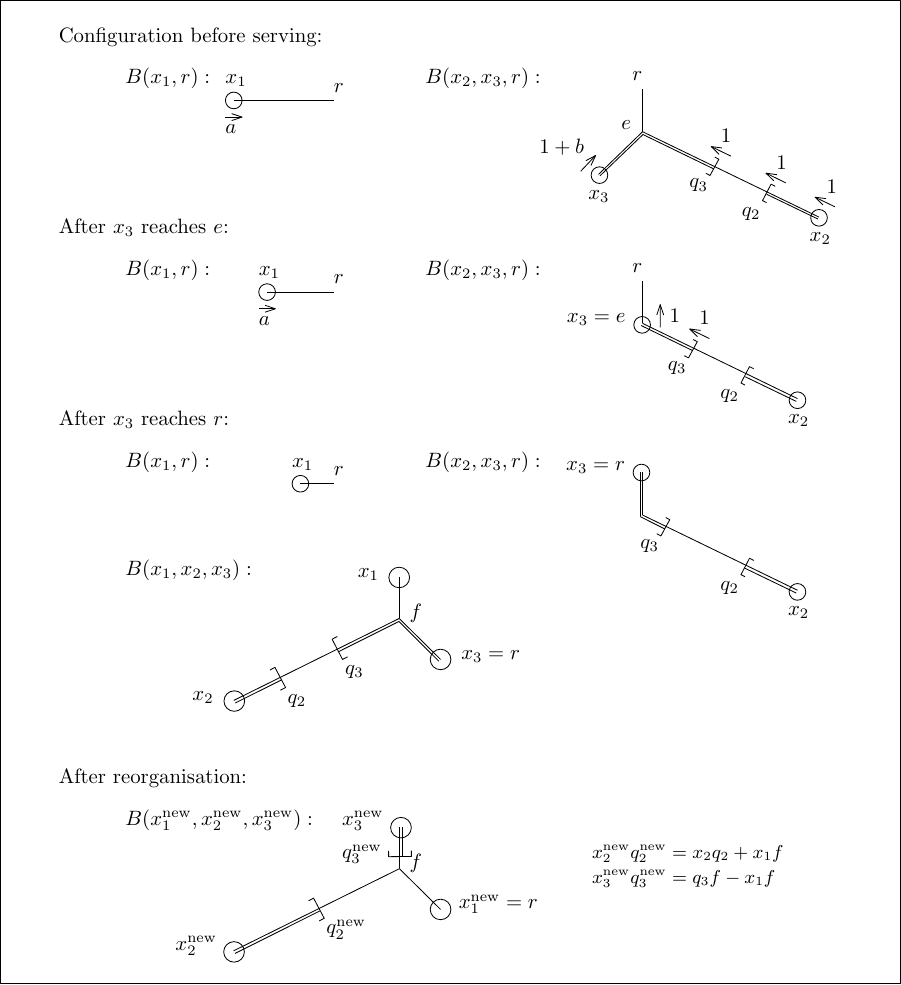}
    \caption{Example of \ouralg{} serving a request. Since the bridges and tripods change over the course of serving a request, each diagram displays the current state of each bridge or tripod at that point in the algorithm.} 
    \label{fig:example}
\end{figure}

\newpage
\bibliographystyle{alpha}
\bibliography{references}

\end{document}